\documentclass[%
 reprint,
 amsmath,amssymb,
 aps,
pra,
]{revtex4-2}

\usepackage{graphicx}% Include figure files
\usepackage{dcolumn}% Align table columns on decimal point
\usepackage{amsthm}
\usepackage{enumitem}
\usepackage{bm}% bold math
\usepackage{hyperref}% add hypertext capabilities
\usepackage{physics}
\usepackage{xcolor}
\usepackage{graphicx}

\newtheorem{theorem}{Theorem}
\newtheorem{lemma}{Lemma}

\newcommand{\cB}{\mathcal{B}}
\newcommand{\cD}{\mathcal{D}}
\newcommand{\cE}{\mathcal{E}}

\newcommand{\cH}{\mathcal{H}}
\newcommand{\cM}{\mathcal{M}}

\newcommand{\cT}{\mathcal{T}}

\begin{document}

\preprint{APS/123-QED}

\title{Theory of mirror benchmarking and demonstration on a quantum computer}

\author{Karl Mayer}
\author{Alex Hall}
\author{Thomas Gatterman}
\author{Si Khadir Halit}
\author{Kenny Lee}
\author{Justin Bohnet}
\author{Dan Gresh}
\author{Aaron Hankin}
\author{Kevin Gilmore}
\author{Justin Gerber}
\author{John Gaebler}

\affiliation{Honeywell Quantum Solutions}
\date{\today}

\begin{abstract}
A new class of protocols called mirror benchmarking was recently
proposed to measure the system-level performance of quantum
computers.
These protocols involve circuits with random sequences of gates followed by mirroring,
that is, inverting each gate in the sequence.
We give a simple proof that mirror
benchmarking leads to an exponential decay of 
the survival probability with sequence length,
under the uniform noise assumption,
provided the twirling group forms a 2-design.
The decay rate is determined by a quantity that 
is a quadratic function of the error channel,
and for certain types of errors is equal to the unitarity.
This result yields a new method for estimating the
coherence of noise.
We present data from mirror benchmarking experiments
run on the Honeywell System Model H1.
This data constitutes a set of performance curves,
indicating the success probability for random
circuits as a function of qubit number and circuit depth.

\end{abstract}

\maketitle

\section{Introduction}\label{sec:level1}

Quantum computers are expected to someday solve
problems that are intractable for classical computers,
but what are the capabilities of quantum computers today?
How should the performance of quantum computers be measured?
Perhaps surprisingly, there is no agreed upon answer to these questions,
though a number of metrics have been proposed.
One such metric, the quantum volume (QV), is given by
the maximum number of qubits for which a class of random 
circuits produces a distribution that passes a certain statistical test~\cite{Cross2019}.
Other groups have proposed to use textbook algorithms such
as Bernstein-Vazirani or Grover search,
which ideally produce deterministic outputs, and report the
algorithmic success probability~\cite{Wright2019}.
Other proposals, inspired by randomized benchmarking (RB)~\cite{Knill2008, Magesan2011},
employ random Clifford circuits of varying length, which can be inverted
in order to measure the survival probability~\cite{Proctor2019}.
There are advantages and disadvantages to each of these
system-level benchmarks,
which have been discussed elsewhere~\cite{Blume-Kohout2020}.

In this paper, we focus on a protocol referred to
as mirror benchmarking, of which different versions were proposed in~\cite{Proctor2020}.
In the variation we consider, the first half of a circuit consists
of a sequence of randomly chosen layers,
and the second half applies the inverse of each
layer in reverse order,
so that the final state is ideally equal to the initial state.
Additionally, random Pauli gates
can be inserted between layers to remove coherent effects,
similar to the randomized compiling procedure~\cite{Wallman2016}.
The output of mirror benchmarking is a plot
of the survival probability as a function of
the sequence length. This can be
useful to a quantum computer user as it provides a 
direct measure of circuit performance for a given qubit number and circuit depth.

The authors of~\cite{Proctor2020} provide substantial motivation for why
mirror benchmarking is a useful benchmark.
They also advocate for the use of mirroring broadly,
to benchmark any circuit of interest,
rather than just random circuits as described here.
Restricting our attention to random circuits,
our main contribution is a formula for how the
survival probability decays as a function of sequence length,
in analogy to RB protocols.
We show that under assumptions commonly made in RB research,
the survival probability decays exponentially.
Under some additional assumptions on the noise,
the decay parameter is 
the unitarity of the average error per circuit layer.
The unitarity was introduced in~\cite{Wallman2015} and is a measure of the coherence of noise.
A consequence of this result is that besides being a useful
benchmark for system-level performance,
mirror benchmarking provides a new method to estimate
the unitarity of noise, for instance of a two-qubit gate.
This requires fewer circuits than the protocol of~\cite{Wallman2015}.

We implement mirror benchmarking
on the Honeywell System Model H1 quantum computer.
In our implementation, 
each layer consists of a random single-qubit Clifford gate on each qubit,
followed by a native two-qubit gate applied to each qubit pair
in a random fully-connected pairing of qubits.
We run mirror benchmarking on $n=6,8,10$ qubits.
For the largest circuits run in this experiment,
with $n=10$ and sequence length $L=16$, corresponding to a two-qubit circuit depth of $32$
and $160$ total two-qubit gates,
we obtain an average survival probability of $0.344(18)$.

The outline of this paper is as follows.
In Sec.~\ref{sec: 2} we give a simple proof
that the survival probability in mirror
benchmarking decays with a rate determined by the unitarity.
In Sec.~\ref{sec: 3} we address the question
of how accurately the direct sampling of gates used in our experiment approximates a
true uniform sampling over a 2-design.
In Sec.~\ref{sec: 4} we describe our
implementation of mirror benchmarking on the 
Honeywell System Model H1 quantum computer and present the results.
Finally, Sec.~\ref{sec: 5} contains a discussion and conclusion.

\section{Mirror benchmarking survival probability}\label{sec: 2}

In a mirror benchmarking protocol,
a sequence of unitaries $g_1,\dots, g_L$,
followed by their inverses $g_L^{-1}, \dots, g_1^{-1}$,
is applied to a noisy initial state,
given by density matrix $\rho$.
For the present discussion, the $g_i$ may be individual gates, 
or more complicated sequences of gates,
but we assume that they are sampled uniformly at random from a group $G$ that forms a unitary 2-design~\cite{Dankert2009}.
Let $\cE$ be a constant error channel associated with each unitary in the sequence.
We assume that $\cE$ is unital, that is, $\cE(I)=I$,
but we discuss how to relax this assumption in the appendix.
We also assume that the error channel per unitary 
during the inversion half of the circuit
is given by $\cE^{\dagger}$, the dual of $\cE$,
which is defined by $\langle\langle\cE^{\dagger}(M)|N\rangle\rangle=\langle\langle M|\cE(N)\rangle\rangle$, where $\langle\langle M| N\rangle\rangle:=\Tr(M^{\dagger}N)$.
We postpone a discussion of this assumption to the end of this section.

At the end of the circuit a measurement is performed 
that is described by a noisy POVM with elements $\{E, I-E\}$.
The average survival probability for a sequence of length $L$ is then given by
\begin{equation}
    p(L) = \frac{1}{|G|^L}\sum_{g_1,\dots, g_L}\langle\bra{E}\cE^{\dagger} g_1^{-1}\cdots\cE^{\dagger} g_L^{-1}\cE g_L\cdots\cE g_1\ket{\rho}\rangle,
\end{equation}
where by abuse of notation we use $g$ to denote both
a unitary and its corresponding superoperator.
We can absorb the final $\cE^{\dagger}$ into the noisy POVM element $E$,
and after defining $g_{i'}= g_i\cdots g_2 g_1$
and relabling $g_{i'}\mapsto g_i$,
the survival probability can be written as
\begin{multline}\label{eq: surv prob 2}
    p(L) = \frac{1}{|G|^L}\sum_{g_1,\dots g_L\in G}\langle\bra{E}(g_1^{-1}\cE^{\dagger} g_1)
    \cdots(g_{L-1}^{-1}\cE^{\dagger}g_{L-1})\\
    (g_L^{-1}\cE g_L)\cdots(g_1^{-1} \cE g_1)\rangle\ket{\rho}\rangle.
\end{multline}
An important quantity is the twirl of $\cE$ over $G$, which is defined as
\begin{equation}
    \cE_T=\frac{1}{|G|}\sum_{g\in G}g^{-1}\cE g.
\end{equation}
A useful fact is that $\cE_T$ commutes with $g$ for all $g\in G$,
which can be shown as follows:
\begin{align}\label{eq: twirl prop}
    g\,\cE_T &= \frac{1}{|G|}\sum_{h\in G}gh^{-1}\cE h\notag\\
    &= \frac{1}{|G|}\sum_{h\in G}(hg^{-1})^{-1}\cE (hg^{-1})g\notag\\
    &= \frac{1}{|G|}\sum_{h'\in G}(h')^{-1}\cE h'g\notag\\
    &= \cE_T g,
\end{align}
where in the third line we set $h'=hg^{-1}$,
and used the group transitivity property to
change the index of summation from $h$ to $h'$.

Summing over $g_L$ in Eq.~\eqref{eq: surv prob 2},
and using Eq.~\eqref{eq: twirl prop}, the central terms in the product in Eq.~\eqref{eq: surv prob 2} are then
\begin{equation}\label{eq: surv prob terms}
    (g_{L-1}^{-1}\cE^{\dag}g_{L-1})\,\cE_{T}\,(g_{L-1}^{-1}\cE g_{L-1}) = g_{L-1}^{-1}\cE^{\dag}\cE_T\cE g_{L-1}.
\end{equation}

Summing the right hand side of Eq.~\eqref{eq: surv prob terms} over $g_{L-1}$,
\begin{equation}
\frac{1}{|G|}\sum_{g_{L-1}} g_{L-1}^{-1}\cE^{\dag}\cE_T\cE g_{L-1} = (\cE^{\dag}\cE_T\cE)_T.
\end{equation}
Continuing in this way,
the general pattern for the survival probability at sequence length $L$ can be described recursively as follows.
Let $\cT_l$ for $l=1,2,\dots$ be the sequence of
operators defined by $\cT_1=\cE_T$ and $\cT_{l+1}=(\cE^{\dagger}\,\cT_l\,\cE)_T$. Then
\begin{equation}\label{eq: surv prob}
    p(L) = \langle\bra{E}\cT_L\ket{\rho}\rangle.
\end{equation}

We will derive an expression for $\cT_L$. First we
recall some facts about twirls over 2-designs~\cite{Dankert2009}.
Let $\cB(\cH)$ denote the vector space of linear operators
on Hilbert space $\cH=\mathbb{C}^d$. Let $V_1=\mathrm{span}\{I\}$ and $V_2=\{A\in\cB(\cH): \Tr(A)= 0\}$,
and note that $\cB(\cH)=V_1\bigoplus V_2$.
Let $\cM:\cB(\cH)\to\cB(\cH)$
and suppose that $G$ is a 2-design. Then the twirl
of $\cM$ over $G$ is a linear combination of
two projectors:
\begin{equation}\label{eq: 2 projectors}
    \cM_T = a\Pi_1 + b \Pi_2,
\end{equation}
where $\Pi_1$ and $\Pi_2$ are the projectors onto $V_1$ and $V_2$, respectively
(A proof of this fact that does not rely on
representation theory is found in~\cite{Nielsen2002}).
If $\cM$ is trace-preserving (TP), then $\cM_T$ is also TP, and therefore
\begin{equation}
    \Tr(I)=\Tr(\cM_T(I))=a\Tr(\Pi_1 I)=a\Tr(I),
\end{equation}
which implies that $a=1$.
Let $D=\dim(V_2)=d^2-1$.
The quantity $b$ is given by
\begin{align}\label{eq: b formula}
    b &= \frac{1}{D}\Tr(\Pi_2\cM_T)\notag\\
    &= \frac{1}{D}\frac{1}{|G|}\sum_{g\in G}\Tr(g\Pi_2 g^{-1}\cM)\notag\\
    &= \frac{1}{D}\Tr(\Pi_2\cM),
\end{align}
where in the last equality we used the
fact that $\Pi_2$ commutes with the action of $G$.
Given an error channel $\cE$, define
\begin{equation}
    f(\cE) = \frac{1}{D}\Tr(\Pi_2\cE).
\end{equation}
The quantity $f(\cE)$ is related to $F(\cE)$, the
process fidelity (also called entanglement fidelity) with respect to the identity~\cite{Nielsen2002},
according to
\begin{equation}
    1-f(\cE) = \frac{d^2}{D}\big(1-F(\cE)\big).
\end{equation}
We will simply write $f$ when the channel $\cE$
is clear from context.
By Eqs.~\eqref{eq: surv prob}-\eqref{eq: b formula}, we have
\begin{equation}\label{eq: twirl}
    \cE_T = \Pi_1 + f\Pi_2.
\end{equation}
The unitarity of $\cE$, which was introduced in~\cite{Wallman2015}, is given by
\begin{equation}\label{eq: unitarity}
    u = \frac{1}{D}\Tr(\Pi_2\cE^{\dagger}\Pi_2\cE).
\end{equation}
We are now ready to prove a formula
for the survival probability.
\begin{lemma}
Let $\cT_l$ for $l=1,2,\dots$ be a sequence of
operators defined by $\cT_1=\cE_T$, and $\cT_{l+1}=(\cE^{\dagger}\cT_l\cE)_T$. Then for all $l$,
\begin{equation*}
    \cT_l = \Pi_1 + fu^{l-1}\Pi_2.
\end{equation*}
\end{lemma}
\begin{proof}
We proceed by induction. The base case is given by Eq.~\eqref{eq: twirl}.
Suppose the statement is true for a particular $l$.
Then
\begin{align}
    \cT_{l+1} &= (\cE^{\dagger}(\Pi_1 + fu^{l-1}\Pi_2)\cE)_T\notag\\
    &= \Pi_1 + \frac{1}{D}\Tr\big(\Pi_2(\cE^{\dagger}(\Pi_1+fu^{l-1}\Pi_2)\cE)\big)\Pi_2,
\end{align}
where in the second equality we used Eqs.~\eqref{eq: 2 projectors} and~\eqref{eq: b formula}
with $\cM=\cE^{\dagger}(\Pi_1 + fu^{l-1}\Pi_2)\cE$.
Since $\cE$ is trace preserving,
$\Pi_1\cE\Pi_2=0$, from which it follows that
\begin{equation}
    \Tr(\Pi_2\cE^{\dagger}\Pi_1\cE)=\Tr(\cE^{\dagger}\Pi_1\cE\Pi_2)=0.
\end{equation}
Therefore,
\begin{align}
    \cT_{l+1} &= \Pi_1 + \frac{1}{D}fu^{l-1}\Tr(\Pi_2\cE^{\dagger}\Pi_2\cE)\Pi_2\notag\\
    &= \Pi_1 + fu^l\Pi_2.
\end{align}
\end{proof}
\begin{theorem}
The survival probability for mirror benchmarking is given by
\begin{equation}
    p(L) = A f u^{L-1} + B,
\end{equation}
for constants $A$ and $B$
that depend only on the state prep and measurement.
\end{theorem}

\begin{proof}
This follows directly from Lemma 1 and Eq.~\eqref{eq: surv prob} with
$A=\langle\bra{E}\Pi_2\ket{\rho}\rangle$
and $B=\langle\bra{E}\Pi_1\ket{\rho}\rangle$.
\end{proof}

A few comments are in order. The assumption that
the error channel during the inversion half of the 
circuit is $\cE^{\dag}$ implies that a unitary error
$U$ on $g_i$ will correspond to error $U^{-1}$ on $g_i^{-1}$.
This may not always hold but often can be enforced by a suitable compilation procedure.
For example, if $g_i=e^{-iZZ\theta}$ is subject to
a control error in $\theta$,
then $g_i^{-1}$ may be implemented by $Pe^{-iZZ\theta}P$, for any Pauli $P$ that anticommutes with $ZZ$.
If we relax this assumption,
and allow for an arbitrary error channel $\cE_{inv}$ on the inverse gates,
then the unitarity $u$ in Th.~1 is replaced with $\frac{1}{D}\Tr(\Pi_2\cE_{inv}\Pi_2\cE)$.

In the limit of perfectly coherent errors, that is $u=1$,
the survival probability does not decay at all,
since the coherent errors during the second half of the circuit
under these assumptions
exactly cancel those from the first half.
At the other extreme, when the error is a depolarizing channel,
$\cE=\Pi_1 + f\Pi_2$ and therefore $u=f^2$ and $p(L)=Af^{2L-1}+B$.
It is only in this special case that the decay parameter also extracts the fidelity of $\cE$.
If $\cE$ is a stochastic Pauli channel,
which we assume in our experimental implementation by the use of Pauli randomization in the protocol,
then $\cE=\cE^{\dag}$ is diagonal in the Pauli basis
of the superoperator representation.
In this case, the unitarity is no longer a measure of coherent errors, but is still defined by Eq.~\eqref{eq: unitarity}.
The process fidelity then is bounded in terms of the unitarity according to
\begin{equation}\label{eq: fid bounds}
    \frac{1}{d^2}\big(1+D\,u\big)\le F(\cE)\le \frac{1}{d^2}\big(1+D\sqrt{u}\big).
\end{equation}
A derivation of this bound is given in the appendix.

\section{Direct circuit sampling}\label{sec: 3}

Theorem 1 of the previous section assumed that the
group $G$ is a 2-design.
However, in the experiment described in the next section,
each layer unitary $g_i$ is sampled from a 
generating set for the $n$-qubit Clifford group $C_n$,
rather than from the full group.
This is done to avoid the need to compile a random
Clifford group element, and so that the sequence lengths of the decay curve more directly correspond
to two-qubit gate depth.
Thus, our implementation of mirror benchmarking is
analogous to the direct benchmarking protocol of~\cite{Proctor2019}.
Here, we justify the claim that Th.~1 should still approximately apply.

The problem of RB using a generating set for the 
Clifford group was studied in~\cite{meier2013}
as a Markov process on the Cayley graph of the Clifford group.
The $i$-th element of the sequence, $g_i$,
approximately samples from the group if an initial distribution
over the generating set diffuses to an approximately uniform distribution over the full graph.
We take a different approach, and simulate the distribution of unitaries $U^{(L)}=g_L\cdots g_2 g_1$
generated by sequences of length $L$,
with the gate set used in our experiment and
described in Sec.~\ref{sec: 4}.
A set of unitaries $\{U_i\}_i$ is a $t$-design if and 
only if
\begin{equation}\label{eq: t-design}
    \Phi_t:=\mathbb{E}_i\abs{\mathrm{Tr}(U_i)}^{2t}=t!\,,
\end{equation}
where the quantity on the left hand side is called the
$t$-th frame potential~\cite{Zhu2017}.
We numerically estimate the $2^{\mathrm{nd}}$ frame potential $\Phi_2$
for the set $\{U^{(L)}_i\}_i$ generated by
mirror benchmarking circuits of sequence length $L=2,4,\dots,16$ and qubit numbers $n=4,6,8$.
The results are shown in Fig.~\ref{fig 2}.
For each data point, we generated 10,000 random circuits 
and computed the quantity in Eq.~\eqref{eq: t-design},
using the standard error for the error bars.

The plot shows convergence to the 2-design value,
with faster convergence for smaller qubit number.
For short sequence lengths, typically $L<n$,
the gate set for direct mirror benchmarking circuits
fails to approximate a 2-design, and therefore fitting 
$p(L)$ to a single exponential decay should be 
understood as a heuristic benchmark, rather than an unbiased estimate of the unitarity.
An interesting and probably challenging open problem
would be to bound the error in $p(L)$ given by
Th.~1 as a function of the deviation from 2-design as measured by $\Phi_2$.

\begin{figure}[ht]
\centerline{\includegraphics[scale=0.55]{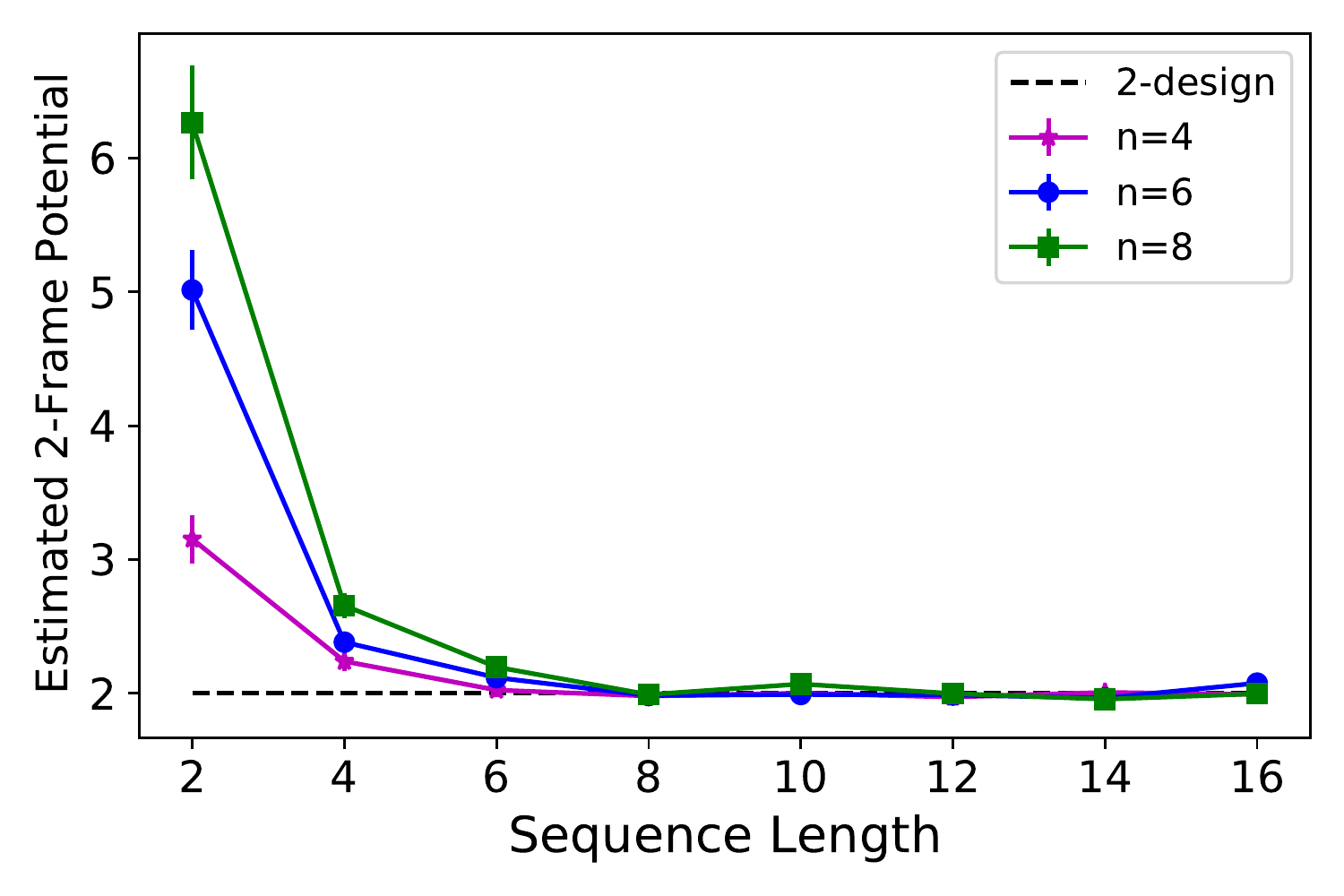}}
 \caption{Plot of $\Phi_2$ defined by Eq.~\eqref{eq: t-design} versus sequence length $L$ for $n=4,6,8$.
 For a $2$-design, $\Phi_2=2$. Each data point is
 estimated as the average over 10,000 random unitaries generated by the direct gate
 set described in Sec.~\ref{sec: 4}.}
 \label{fig 2}
\end{figure}

\section{Experiment}\label{sec: 4}

We implement mirror benchmarking on the
Honeywell System Model H1 quantum computer.
This system contains 10 trapped-ion qubits in a linear QCCD architecture described in~\cite{Pino2021, Tobey2020}.
The native two-qubit gate is the phase-insensitive
M\o{}lmer-S\o{}renson gate~\cite{Sorenson2000, Lee2005},
given by
\begin{equation}
    U_{zz} = e^{-iZZ\pi/4}.
\end{equation}
A circuit diagram of our implementation is shown in Fig.~\ref{fig: circuit diagram}.
Each $g_i$ is a circuit layer
with a random single-qubit Clifford gate applied to each qubit
followed by a native two-qubit gate applied to
each qubit pair in a random fully-connected
pairing of qubits.
Finally, we include Pauli randomization to prevent coherent errors
from systematically producing higher survival probabilities.
Random Pauli gates that multiply 
to the identity are inserted between each layer.
The second Pauli gate in each product is pushed through the subsequent layer using the 
commutation relations and combined with the next single-qubit gate.

\begin{figure}[ht]
\centerline{\includegraphics[scale=0.35]{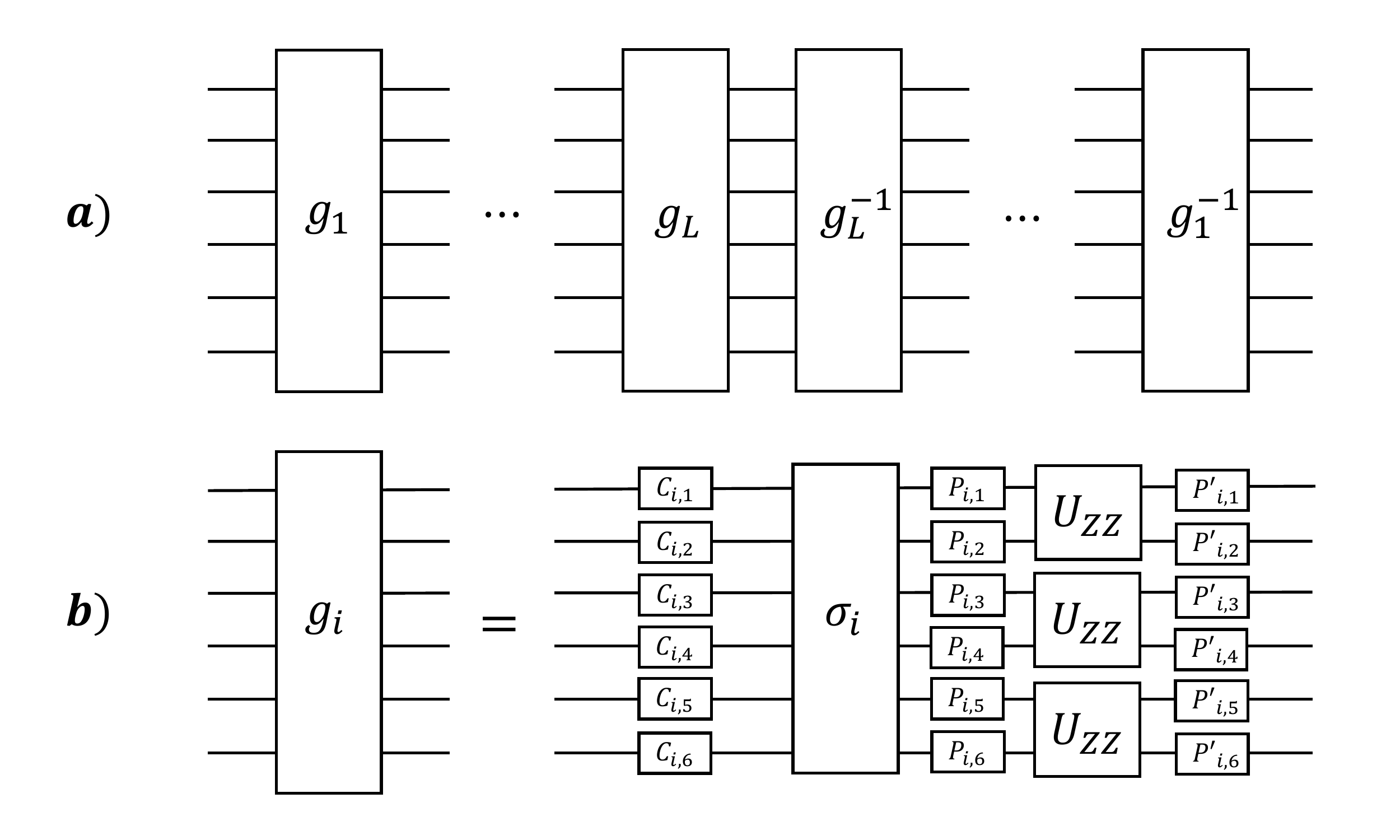}}
 \caption{
 Mirror benchmarking circuit diagram.
 \textbf{a)} General form for mirror benchmarking with $L$ layers,
 shown for $n=6$.
 \textbf{b)} The $i$-th layer unitary used in our implmentation.
 Random single-qubit Clifford gates $C_{i,j}$ are applied to qubit $j$.
 A random permutation $\sigma_i\in S_n$ is applied,
 and the native $U_{ZZ}$ gate is then applied to each qubit pair.
 Random Paulis $P_{i,j}$ are inserted before and pushed
 through the $U_{ZZ}$ gates to leave the circuit unitary
 unchanged. The Paulis $P'_{i,j}$ are combined with the
 subsequent round of Cliffords before execution on actual hardware.
 }
 \label{fig: circuit diagram}
\end{figure}

The results of our experiment are shown in Fig.~\ref{fig: decays}.
The figure shows the average $p(L)$ versus $L=4,8,12,16$ and with qubit number $n=6,8,10$.
We only use an even number of qubits,
since increasing to odd $n+1$ does not increase
the number of two-qubit gates in our circuits,
which are the leading source of circuit error.
For each sequence length, we sample 10 circuits,
and repeat each circuit for 100 shots.
We insert random Pauli gates before the measurement,
so that the constant $B$ can be fixed to $1/2^{n}$,
as described in~\cite{Harper2019}.
The error bars are computed by a bootstrapping procedure adapted from~\cite{meier2013}.
For each sequence length, the random circuits are resampled with 
replacement and then a parametric bootstrap is performed on the observed probability distribution for each circuit.
The unitarity estimate at $n=10$ is $u=0.938(4)$,
corresponding to process fidelity bounds of
$0.938(4)\le F(\cE)\le 0.969(2)$ per random fully-connected layer with $5$ two-qubit gates.

\begin{figure}[ht]
\centering
\includegraphics[scale=0.55]{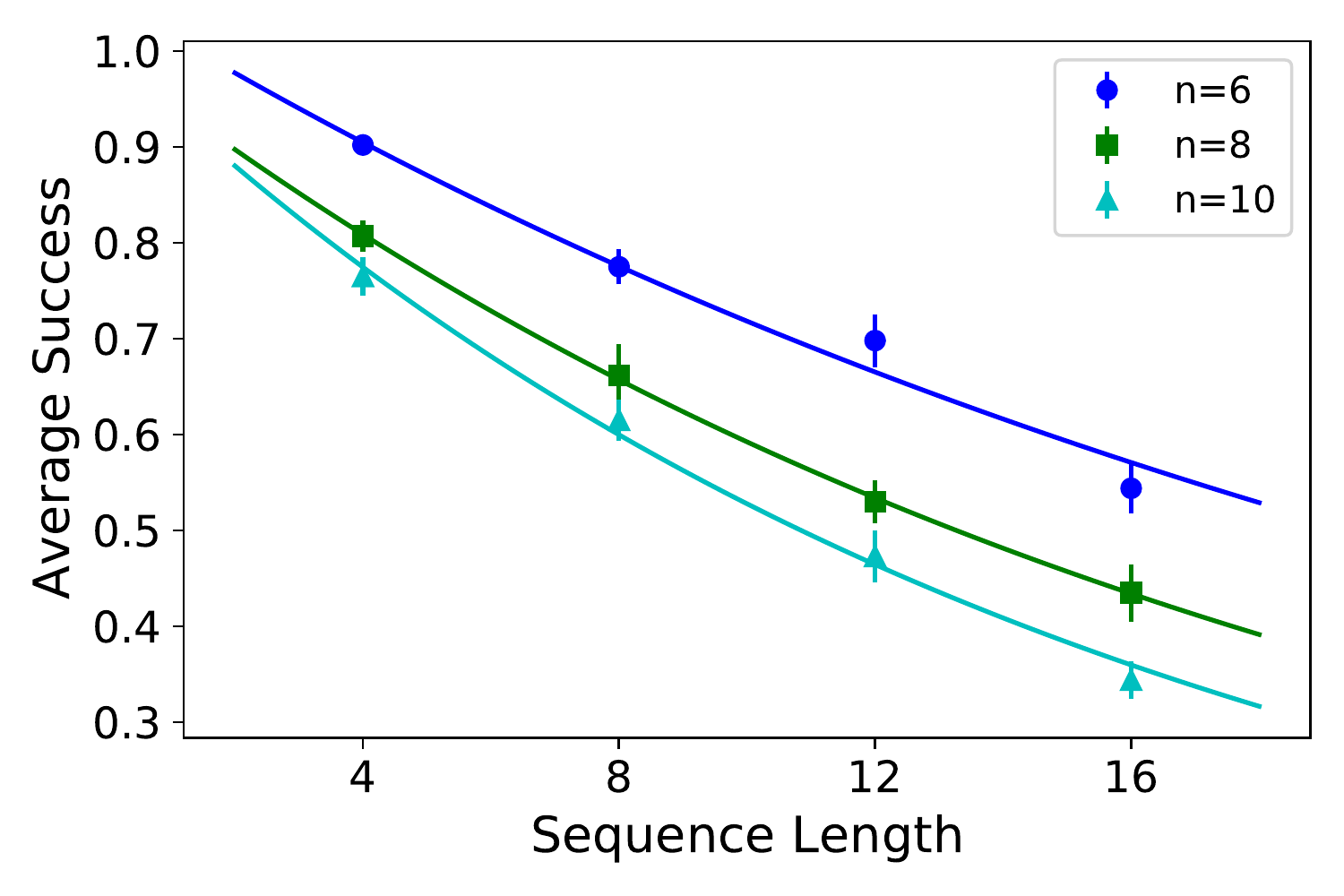}
\centering
    \begin{tabular}{ |c|c|c|c| }
    \hline
    $n$ & $6$ & $8$ & $10$\\
    \hline
    Unitarity & $0.962(3)$ & $0.949(5)$ & $0.938(4)$\\
    \hline
    $F_{lower}$ & $0.962(3)$ & $0.949(5)$ & $0.938(4)$\\
    \hline
    $F_{upper}$ & $0.981(2)$ & $0.974(3)$ & $0.969(2)$\\
    \hline
    \end{tabular}
 \caption{Mirror benchmarking data plotting $p(L)$ versus $L$
 for qubit numbers $n=6,8,10$, with
 10 circuits per sequence length,
 and 100 shots per circuit.
 The data was taken on the Honeywell System Model H1.
 The table gives the unitarity estimate obtained by
 best-fitting to $p(L)=Au^{L-1}+1/2^n$,
 as well as the lower and upper bounds on the process fidelity 
 given by Eq.~\eqref{eq: fid bounds}.}
 \label{fig: decays}
\end{figure}

% \begin{figure}[ht]
% \centerline{\includegraphics[scale=0.55]{plot.png}}
%  \caption{Simulation of $p(L)$ versus $L$
%  for qubit numbers $n=6,8,10$,
%  and 10 circuits per sequence length.
%  The simulation used a depolarizing error model
%  with parameter 0.01 for each two-qubit gate
%  and a measurement error of 0.002 per qubit.}
%  \label{fig: decays}
% \end{figure}

To study the accuracy and precision of the unitarity estimate,
we simulate 100 different mirror benchmarking experiments
for $n=6,8,10$ with depolarizing error on the two-qubit gates.
For each simulated experiment, the depolarizing parameter
was chosen uniformly at random between $0.0$ and $1.0\times10^{-2}$.
Figure~\ref{fig: scatter} shows a scatter plot of the
estimated unitarity versus depolarizing parameter,
with the true unitarity plotted for comparison.
The true unitarity for circuit layers consisting of $n/2$ two-qubit gates with uniform depolarizing error is given
by a formula that we derive in the appendix.
For a single point, the estimation error is the difference between
the estimated and true unitarities.
The mean of the estimation errors of the simulated data points is
less than $2\times10^{-4}$, and the
standard deviations are $1.5,2.2,2.6\times10^{-3}$ for $n=6,8,10$, respectively.
This indicates that for this error model,
mirror benchmarking gives an unbiased estimate of the unitarity,
to within the numerical precision of the simulation.

\begin{figure}[ht]
\centering
\includegraphics[scale=0.55]{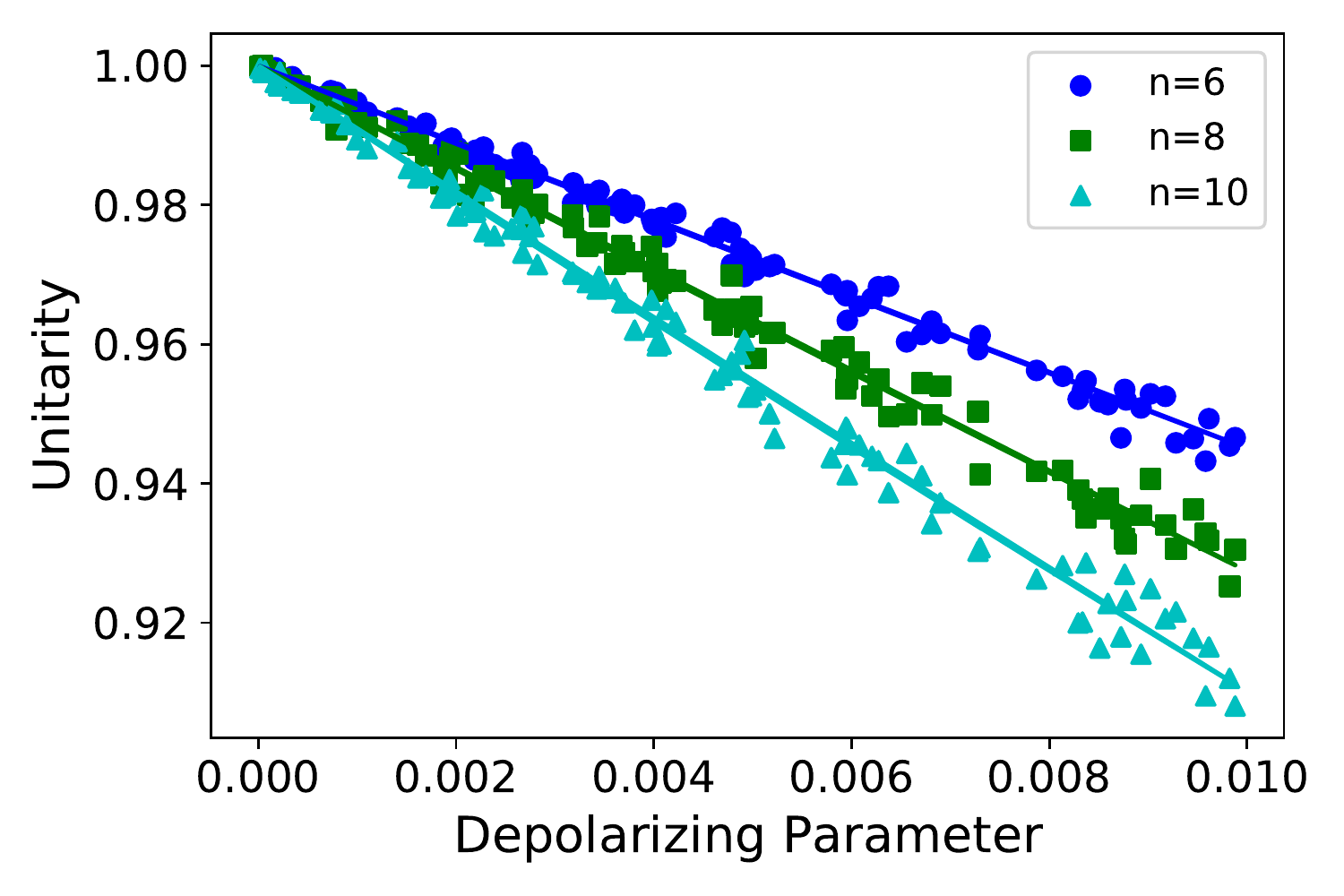}
\centering
Unitarity Estimation Error
  \begin{tabular}{ |c|c|c|c| }
     \hline
     $n$ & $6$ & $8$ & $10$\\
     \hline
     Mean & $1.2\times10^{-4}$ & $1.2\times10^{-5}$ & $1.8\times10^{-4}$\\
     \hline
     Std. & $1.5\times10^{-3}$ & $2.2\times10^{-3}$ & $2.6\times10^{-3}$\\ 
     \hline
    \end{tabular}
 \caption{ Scatter plot of unitarity estimates
 from 100 simulated mirror benchmarking experiments
 versus two-qubit depolarizing parameter.
 The solid lines plot the true unitarity.
 Each experiment used sequence lengths $L= 4, 8, 12, 16$ and $10$ circuits per sequence length.
 }
 \label{fig: scatter}
\end{figure}

\section{Discussion}\label{sec: 5}

We have shown that for mirror benchmarking with random circuits,
under some assumptions
the survival probability decays exponentially according
to the unitarity of the average error channel per layer.
This justifies the use of mirror benchmarking
as a system-level test of quantum computer performance and leads to a number of applications.
It is well known that coherent errors can cause large quantum circuits to perform worse
than would be predicted from the fidelities of
single and two-qubits gates alone~\cite{Kueng2016}.
A mirror benchmarking experiment could be performed
with and without the use of randomized compiling.
If the noise is partially coherent, then
the survival probability and extracted 
unitarity will be higher without randomized compiling,
indicating that randomized compiling should be used 
for real quantum circuits of interest.
Apart from providing a system-level benchmark,
mirror benchmarking also reduces to a protocol for
estimating the unitarity of gate errors,
when performed on just one or two qubits.
This is an alternative to the protocol in~\cite{Wallman2015},
which relies on repeating circuits over multiple
measurement settings.
We leave an exploration of the ability for mirror
benchmarking to diagnose coherent errors to a future work.

We conclude by pointing out some limitations of
our present work and open problems.
First, we proved our main result under the uniform 
noise assumption, that is, that $\cE$ is the
same error channel for every layer $g_i$.
This assumption is unrealistic but is commonly made in the research on RB,
and often works well in practice.
We speculate that this assumption
can be relaxed, as was shown for RB in~\cite{Wallman2018}.
Also, we do not present a statistical analysis
for the unitarity estimate (i.e. a confidence interval as a function of sample and circuit complexity) in this work.
We believe that methods from~\cite{Wallman2014} and~\cite{Harper2019} can be adapted to our
setting,
but this is an open problem.
Finally, we remark that the layer unitaries
in mirror benchmarking in general can be chosen to be whatever one wants.
In particular, one could choose Haar random elements of $SU(4)$ for each qubit pair to define a mirrored version
of the quantum volume (QV) test~\cite{Cross2019}. 
This would solve the problem that QV requires classical simulation and is therefore fundamentally unscalable.
We wonder what the relationship is between the survival
probability and the heavy outcome probability in ensembles of mirror benchmarking and QV circuits, respectively,
and whether mirror benchmarking can be used to extrapolate QV 
beyond the regime of classical simulatability. 

\acknowledgements
We acknowledge Timothy Proctor for a helpful discussion.
We thank Charlie Baldwin and David Hayes and for comments and suggestions.
We thank the entire experimental team at Honeywell Quantum Solutions for making this work possible.

\nocite{*}

\bibliography{library}% Produces the bibliography via BibTeX.

\appendix

\section{Non-unital error channels}

In Sec.~\ref{sec: 2}, we assumed that $\cE$ is unital, but we can relax this 
assumption and recover the result of Th.~1 as follows.
As a CPTP map, $\cE$ can be written as $\cE=\Pi_1 + \cE_n+\cE_u$,
where $\cE_n=\Pi_2\cE\Pi_1$ and $\cE_u=\Pi_2\cE\Pi_2$ are
the non-unital and unital parts of $\cE$, respectively.
If $\cE$ has a non-unital part, then $\cE^{\dagger}$ is not TP. 
In this case, we make the assumption that the error channel
following each inverse layer $g_i^{-1}$ is given by
$\cE':=\Pi_1+\cE_n+\cE_u^{\dagger}$.
That is, the error channel has the same
non-unital part but has a unital part that is dual to the unital part of $\cE$.
The motivation for this assumption is that we expect control errors 
on gates (such as over-rotations) to be inverse to
the control errors on the corresponding inverse gates.
Non-unital errors (such as amplitude damping channels), 
however, are expected to be constant between the front half
and back half of the mirrored circuit.
By replacing $\cE^{\dag}$ with $\cE'$ in Sec.~\ref{sec: 2}
we recover Th.~1.
\newline

\section{Bounds on process fidelity from unitary}

Let $\cE_{ij}=\frac{1}{d}\Tr(P_i \cE(P_j))$
be the matrix elements of the error channel in the Pauli basis of the superoperator representation,
with the convention that $P_0=I$.
The unitarity is given by
\begin{equation}
    u = \frac{1}{D}\sum_{i,j>0} \cE_{ij}^2 \ge \frac{1}{D}\sum_{i>0} \cE_{ii}^2.
\end{equation}
The latter inequality is saturated when $\cE$ is a
stochastic Pauli channel.
Meanwhile, the process fidelity is
\begin{equation}
    F(\cE) = \frac{1}{d^2}\big(1+\sum_{i>0}\cE_{ii}\big).
\end{equation}
From the previous two equations, upper and lower bounds on $F(\cE)$
for stochastic Pauli channels can be obtained.
On one hand, since each $\cE_{ii}\in[0,1]$,
\begin{equation}
    \sum_{i>0} \cE_{ii}^2\le \sum_{i>0}\cE_{ii}.
\end{equation}
On the other hand, minimizing $\sum_{i>0}\cE_{ii}^2$
under the constraint that $\sum_{i>0}\cE_{ii}=c$ gives
$\cE_{ii}=c/D$ for all $i$,
at which point $\sum_{i>0}\cE_{ii}^2=c^2/D$. Therefore,
\begin{equation}
    \frac{1}{D}\big(\sum_{i>0}\cE_{ii}\big)^2\le\sum_{i>0} \cE_{ii}^2.
\end{equation}
Using these two inequalities and
substituting the expressions for $u$ and $F(\cE)$,
after a little algebra we obtain
\begin{equation}
    \frac{1}{d^2}\big(1+D\,u\big)\le F(\cE)\le \frac{1}{d^2}\big(1+D\sqrt{u}\big).
\end{equation}

\section{Unitarity of tensor-product of depolarizing channels}

We derive a formula for the unitarity of $n/2$ parallel 
two-qubit depolarizing channels, with $n$ even.
For a system of dimension $d$, a depolarizing channel
with parameter $p$ acts on a linear operator $X$ as
\begin{equation} 
    \cD(X) = (1-p)X + p\Tr(X)I/d.
\end{equation}
Now fix $d=4$, corresponding to a two-qubit system, and let $N=n/2$.
We are interested in the tensor product channel $\cD^{\otimes N}$.
For $n$-qubit Pauli operator $P=P_1P_2\cdots P_N$, where $P_i$ 
is a two-qubit Pauli operator acting on the $i$-th qubit pair, we have

\begin{widetext}
\begin{align}
\cD^{\otimes N}(P_1\cdots P_N) &= \bigotimes_{i=1}^N\bigg((1-p)P_i + \frac{p}{4}\Tr(P_i)I\bigg)\notag\\
&= \Big( (1-p)^N + (1-p)^{N-1} p\sum_i\delta_{P_i,I} + (1-p)^{N-2}p^2 \sum_{i<j}\delta_{P_i,I}\delta_{P_j,I} + \dots\Big)P_1\cdots P_N\notag\\
&= \Big(\sum_{j=0}^{N-w(P)}\binom{N-w(P)}{j}(1-p)^{N-j}p^j\Big)P_1\cdots P_N,
\end{align}
\end{widetext}
where in the second line we used $\Tr(P_i)=4\delta_{P_i,I}$,
and where $w(P)=|\{i:P_i\ne I\}|$ is the weight of $P$ with respect to the $N$ two-qubit subsystems.
The unitarity is given by
\begin{equation}
    u(\cD^{\otimes N}) = \frac{1}{(4^N-1)4^N}\sum_{P\ne I}\Tr\big(P\cD^{\otimes N}(P)\big)^2. 
\end{equation}

We compute the sum over $P$ by first summing over the weights of $P$.
For a given $w$, there are $15^w\binom{N}{w}$ Paulis of weight $w$,
since there are $\binom{N}{w}$ choices of $w$ subsystems, and $(4^2-1)=15$
different Paulis per subsystem satisfying $P_i\ne I$.
Therefore,

\begin{widetext}
\begin{equation}
    u(\cD^{\otimes N}) = \frac{1}{4^N-1}\sum_{w=1}^N 15^w\binom{N}{w}\Bigg(\sum_{j=0}^{N-w}\binom{N-w}{j}(1-p)^{N-j}p^j\Bigg)^2 .
\end{equation}
This is the equation for the true unitarity that is
plotted in Fig.~\ref{fig: scatter}.
\end{widetext}

\end{document}